\documentclass[reqno,12pt,letterpaper]{amsart}
\usepackage{amsmath,amssymb,amsthm,graphicx,mathrsfs,url}
\usepackage[usenames,dvipsnames]{color}
\usepackage[colorlinks=true,linkcolor=Red,citecolor=Green]{hyperref}
\usepackage{amsxtra}
\usepackage{wasysym} 
\usepackage{graphicx}
\usepackage{subcaption}
\def\arXiv#1{\href{http://arxiv.org/abs/#1}{arXiv:#1}}
\usepackage{array}
\newcolumntype{P}[1]{>{\centering\arraybackslash}m{#1}}

\def\?[#1]{\textbf{[#1]}\marginpar{\Large{\textbf{??}}}}

\def\smallsection#1{\smallskip\noindent\textbf{#1}.}
\let\epsilon=\varepsilon 

\newcommand{\RR}{{\mathbb R}}

\newcommand{\CC}{{\mathbb C}}

\newcommand{\ZZ}{{\mathbb Z}}

\newtheorem{theo}{Theorem}
\newtheorem{prop}{Proposition}[section]

\newtheorem{lemm}[prop]{Lemma}

\numberwithin{equation}{section}

\DeclareMathOperator{\Res}{Res}
\DeclareMathOperator{\Spec}{Spec}

\DeclareMathOperator{\tr}{tr}

\usepackage{scalerel}

\newcommand\reallywidehat[1]{\arraycolsep=0pt\relax%
\begin{array}{c}
\stretchto{
  \scaleto{
    \scalerel*[\widthof{\ensuremath{#1}}]{\kern-.5pt\bigwedge\kern-.5pt}
    {\rule[-\textheight/2]{1ex}{\textheight}} 
  }{\textheight} %
}{0.5ex}\\           
#1\\                 
\rule{-1ex}{0ex}
\end{array}
}

\title{Integrability in the chiral model of magic angles} 

\author{Simon Becker}
\email{simon.becker@math.ethz.ch}
\address{ETH Zurich, 
Institute for Mathematical Research, 
Rämistrasse 101, 8092 Zurich, 
Switzerland}

\author{Tristan Humbert}
\email{tristan.humbert@ens.psl.eu}
\address{ENS Paris, Département de Mathématiques et Applications, 
Rue d'Ulm, Paris, 
France}

\author{Maciej Zworski}
\email{zworski@math.berkeley.edu}
\address{Department of Mathematics, University of California,
Berkeley, CA 94720, USA.}
\begin{document}
\begin{abstract}
Magic angles in the chiral model of twisted bilayer graphene are
parameters for which the chiral version of the Bistritzer--MacDonald
Hamiltonian exhibits a flat band at energy zero.
We compute the sums over powers of (complex) magic angles and use that to show
that the set of magic angles is infinite. We also provide
a new proof of the existence of the first real magic angle, showing also that
the corresponding flat band has minimal multiplicity for the simplest possible choice of potentials satisfying all symmetries. These results
indicate (though do not prove) a hidden
integrability of the chiral model.
\end{abstract}

\maketitle 

\section{Introduction and statement of results}

When two sheets of graphene are stacked on top of each other and twisted, it has been observed that at certain angles, coined the \emph{magic angles}, the composite system becomes superconducting. In this article, we study the chiral limit of the Bistritzer-MacDonald Hamiltonian \cite{BM11,CGG,Wa22} 
\[H(\alpha)=\begin{pmatrix} 0 & D (\alpha)^*\\ D(\alpha) & 0 \end{pmatrix} \text{ with } D(\alpha) = \begin{pmatrix} D_{\bar z}& \alpha U(z) \\ \alpha U(-z) & D_{\bar z} \end{pmatrix} \]
where the parameter $\alpha$ is proportional to the inverse relative twisting angle. After a simple rescaling, the potential is a smooth and periodic function satisfying
 \begin{equation}
 \label{eq:potential}
 U(z+a_\ell ) = \bar \omega U(z), \quad U(\omega z)= \omega U(z), \text{ and } U(\bar z)= \overline{U(z)},
 \end{equation}
 where $\omega = e^{2\pi i/3}$ and $a_\ell = \frac{4}{3} \pi i \omega^\ell.$ The simplest example of such a potential and our canonical choice of $U$ is 
 \begin{equation}
 \label{eq:defU} U_0(z) = \sum_{k=0}^2 \omega^k e^{\frac{1}{2}(z \bar \omega^k - \bar z \omega^k)}.
 \end{equation}
 
 Even though the potential $ U ( z ) $  is only periodic with respect to $\Gamma=4\pi i(\omega \ZZ \oplus\omega^2 \ZZ)$ the first property implies that the matrix potential, and thus $D(\alpha)$, commutes with the translation operator
\begin{equation}
\label{eq:4and11}   \mathscr L_{a } w ( z ) := 
\begin{pmatrix}  \omega^{a_1 + a_2}  & 0 \\
0 & 1  \end{pmatrix} 
w( z +  a ),  \ \ \ a \in \tfrac13 \Gamma  , \end{equation}
where $ w \in \mathbb C^2 $ and $ a = \tfrac 4 3 \pi i( \omega a_1 + \omega^2 a_2 ) , \ \ a_j \in \mathbb Z.$
We note that if $ \Gamma^*$  is the dual (reciprocal) lattice of $ \Gamma $, then 
 $ 3 \Gamma^* $ is the  dual lattice of $ \frac 13 \Gamma$. 

\begin{figure}
\includegraphics[width=10cm]{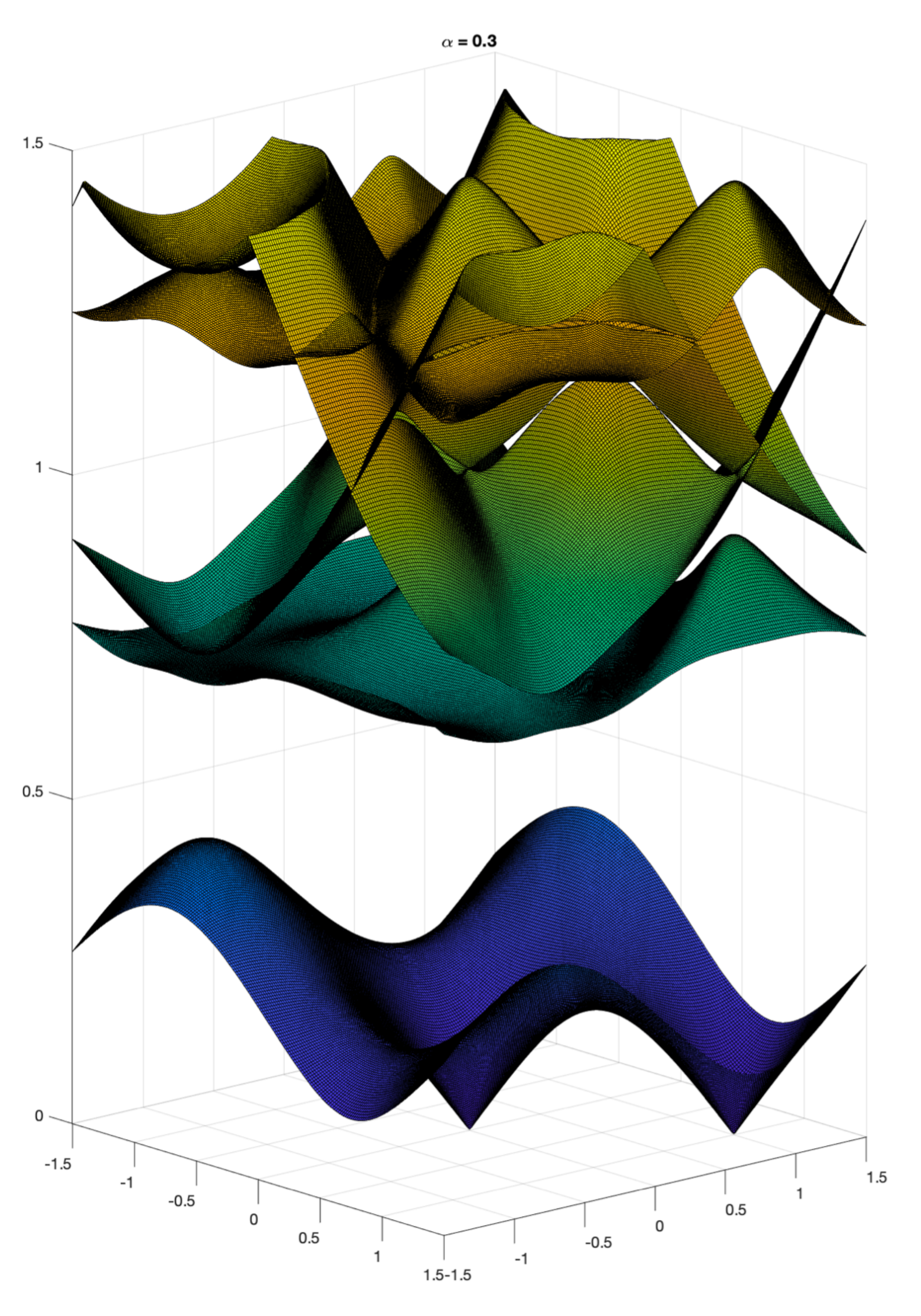}
\caption{\label{f:bands} 
Plots of the first 5 non-negative eigenvalues of
$ H_{\mathbf k } ( 0.3 ) $ acting on $ \mathscr H $ (see \eqref{eq:defHk} and \eqref{eq:defscH}), 
as function of $  k = (k_1 \omega^2 - k_2 \omega)/\sqrt 3 $ in 
in a fundamental cell of 
$ 3 \Gamma^* $,  parametrized by 
$ ( k_1, k_2 ) $
$ |k_j| < \frac32 $. See also \cite[Figure 4]{bhz2} for more information
and comparison with band structure of other models.}
\end{figure}

When moving to functions with values in $ \mathbb C^4 = \mathbb C^2 \times
\mathbb C^2 $ (on which $ H(\alpha ) $
acts) we extend the action of $ \mathscr L_{\mathbf a}$  to an action on each
$\mathbb C^2 $ component. We then consider the Floquet spectrum of
   \begin{equation}
   \label{eq:defHk}   H_{k } ( \alpha ) = \begin{pmatrix} 0 & D(\alpha)^*-\bar k\\ D(\alpha)-k &0 \end{pmatrix}\text{ with } k \in 3\Gamma^*, \end{equation}
   defined by $(H_k(\alpha) -E_j(\alpha,k) )w_j(\alpha,k)=0$, where eigenvalues of positive energy are labelled with $j \ge 1$ in ascending order, as a self-adjoint operator on 
\begin{equation}
\label{eq:defscH} 
\mathscr H := \{ \mathbf v \in L^2 ( \mathbb C/\Gamma ) :
\mathscr L_{\mathbf a } \mathbf v = \mathbf v , \ \ 
\mathbf a \in \tfrac13 \Gamma 
\} ,
\end{equation}
with the domain given by $ \mathscr H \cap H^1 ( \mathbb C/\Gamma   ) $ such that $$\Spec_{L^2 ( \mathbb C ; \mathbb C^4 ) }(H(\alpha))= \bigcup_{k \in \CC} \Spec_{\mathscr H} (H_{k}(\alpha)).$$
 
This Hamiltonian is an effective one-particle model which exhibits perfectly flat bands at magic angles. This appearance of perfectly flat bands in the chiral limit was explained by 
Tarnopolsky, Kruchkov and Vishwanath \cite{magic} with the help of Jacobi theta functions\footnote{As was pointed out to us by Alex Sobolev a similar argument appeared in the 
work of Dubrovin and Novikov \cite{dun} who studied magnetic Hamiltonians on tori.}. An equivalent spectral theoretic characterization of magic angles was then provided in \cite{beta}: if 
we define the following compact Birman-Schwinger operator
\begin{equation}
\label{eq:BS}
T_k = (2D_{\bar z}-k)^{-1} \begin{pmatrix} 0 & U(z) \\ U(-z) & 0 \end{pmatrix}. 
\end{equation}
then (see \cite[Theorem 2]{beta} we have the following equivalence
\cite[\S 2.3]{bhz2}) 
\begin{equation}
\label{eq:specc} 0 \in \bigcap_{k \in \CC} \Spec_{\mathscr H} (H_k(\alpha) ) \  \Longleftrightarrow \
\left\{ \begin{array}{l} 
\alpha^{-1} \in \Spec_{\mathscr H} (T_{k_0}) \\
 \text{ for some }k_0 \in \CC \setminus (3\Gamma^* - \{ 0, i \} ) , \end{array} \right. 
\end{equation} where $ \mathscr H $ is defined in \eqref{eq:defscH}. In other words, the spectrum of $T_{k_0}$ is independent of $k_0 \in \CC \setminus  (3\Gamma^* - \{ 0, i \} )$ and characterizes the values of $\alpha \in \CC$ at which the Hamiltonian exhibits a flat band at zero energy. Since the parameter $\alpha$ is inherently connected with the twisting angle, we shall refer to $\alpha$'s at which \eqref{eq:specc} occurs as \emph{magic} and
denote their set by $ \mathcal A \subset \mathbb C $.

\begin{figure}
\begin{center}
\includegraphics[width=12cm]{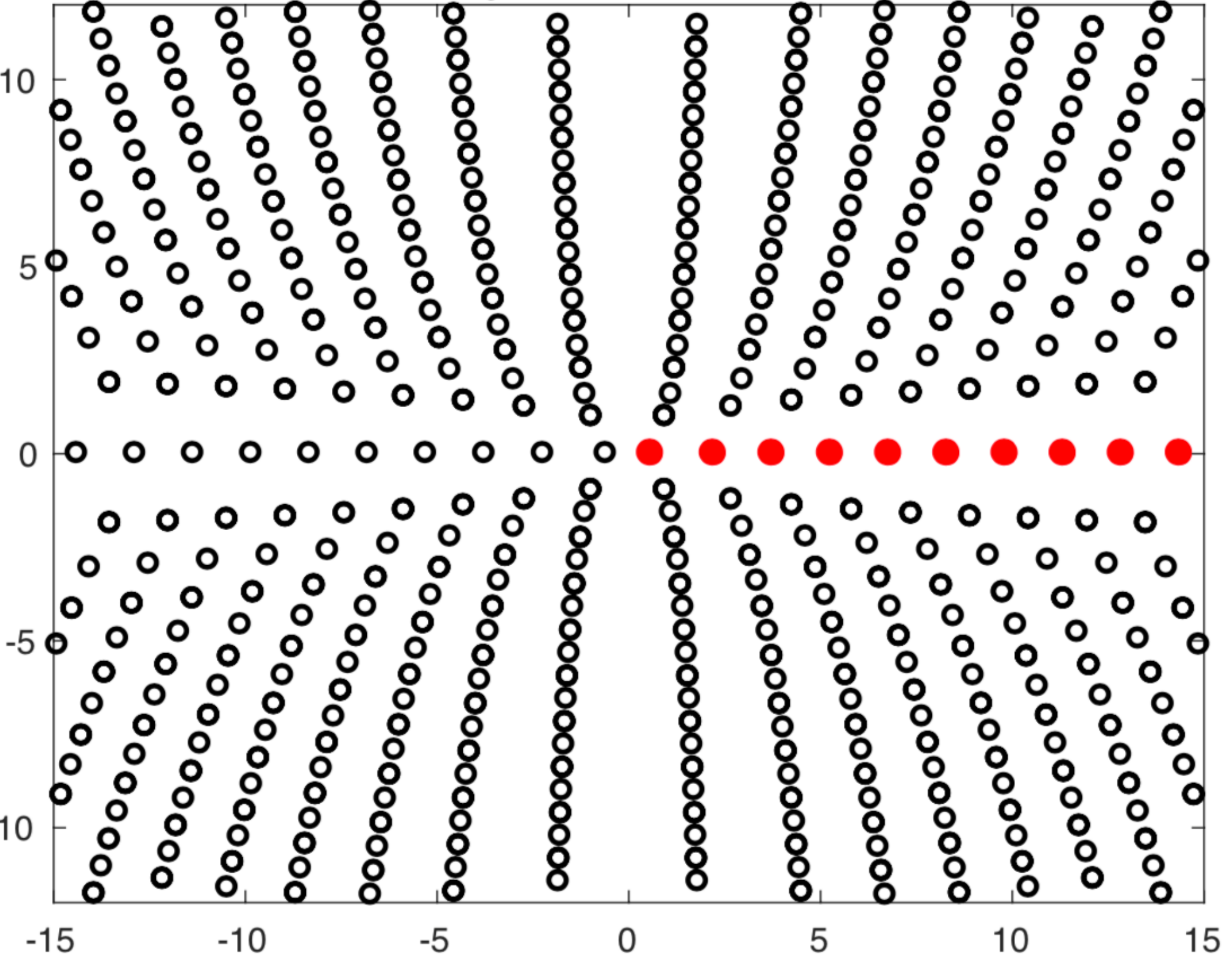}
\end{center}
\caption{\label{f:magic}
The set $ \mathcal A $ of magic $ \alpha$'s for which \eqref{eq:specc} holds, that is, 
the first band is flat. The positive elements of $ \mathcal A $
are the reciprocals
of the ``physically relevant" positive angles. 
Potential \eqref{eq:defU} is responsible for the regularity of the set
which seems to indicate hidden integrability. For more general potentials the
distribution is more complicated -- see \url{https://math.berkeley.edu/~zworski/multi.mp4}
for $ U_\theta ( z ) = (\cos^2 \theta) U ( z ) + 
( \sin^2\theta) \sum_{k=0}^2 \omega^k e^{\bar z \omega^k-  z \bar \omega^k} $
which satisfies  the required symmetries
 \eqref{eq:potential}. The animation also indicates changing multiplicities.}
\end{figure}

The analysis of magic angles is therefore reduced to a spectral theory problem involving a single compact non self-adjoint operator. Since even non-trivial non self-adjoint compact operators do not necessarily have non-zero eigenvalues, the existence of a parameter $\alpha$ at which the Hamiltonian exhibits a flat band at zero energy is non-trivial. In \cite{beta} the existence of such a complex parameter $\alpha \in \CC \setminus \{0\}$ was first concluded by showing that $\tr_{\mathscr H} (T_k^4)  = 8{\pi}/{\sqrt{3}}$ which implied existence of a non-zero eigenvalue\footnote{In \cite{beta} we 
considered the trace on $ L^2 ( \mathbb C/\Gamma; \mathbb C^2 ) $ which gave this 
answer multiplied by $ 9 $.}.
      This result was improved by a computer-assisted proof \cite{lawa} in which Watson and Luskin
       used the complex-analytic characterization of magic angles from \cite{magic} to prove existence of the  first \emph{real} magic angle and obtained explicit bounds on its position.

In this article, we exhibit a general form of traces of powers of $T_k$. 
This suggests a hidden {\em integrability} of the Hamiltonian $ H ( \alpha ) $ 
for potentials satisfying  \eqref{eq:potential}, as all traces exhibits special arithmetic properties. With our current techniques, we do not have explicit control on the full set of traces which would imply a complete understanding of all magic angles.
These are already visible in the regular but evasive structure of the set of of magic
$ \alpha$, $ \mathcal A \subset \mathbb C $ -- see Figure \ref{f:magic}.

\begin{theo}
\label{theo:traces}
For $\ell \ge 2$ and $U=U_0$ with $U_0$ as in \eqref{eq:defU}
\begin{equation}
\label{eq:traces} \tr(T_k^{2\ell})=\sum_{\alpha \in \mathcal A  } \alpha^{-2\ell}=   \frac{\pi}{\sqrt{3}} \, q_{\ell } \ 
\text{ with } q_{\ell} \in \mathbb Q.\end{equation}
\end{theo} 
In addition, we are able to express the rational numbers $q_{\ell} \in \mathbb Q$ in terms of a finite sum involving residues of rational functions which is fully presented in Theorem \ref{traceresult}. A generalization of Theorem \ref{theo:traces} which extends this result to more general potentials $U$ \eqref{eq:potential} is presented in Theorem \ref{rat}.
As we show in \S  \ref{sec:races_to_infinity}, it is already possible to conclude directly from Theorem \ref{theo:traces} that 
\begin{theo}
\label{corr:infinity}
Let $U = U_0$ with $U_0$ as in  \eqref{eq:defU}. There exist infinitely many magic $\alpha$'s, that is,
\[ | \mathcal A | = \infty . \]
\end{theo}
This theorem will follow from the more general Theorem \ref{sec:races_to_infinity} and the observation that by the aforementioned explicit computation $\tr_{\mathscr H} (T_k^4)  = 8{\pi}/{\sqrt{3}}$ for $U=U_0$ there is at lest one complex magic angle. 
We then focus on real magic angles. Since the operator $T_k^2$ is Hilbert-Schmidt, we can use the regularized determinant to study real magic $\alpha.$ Compared with the initial approach proposed in \cite{magic}, this approach has two advantages. Unlike the series expansion in \cite{magic,lawa}, the regularized determinant is an entire function with explicit error bounds in terms of the Hilbert-Schmidt norm. In addition, the Taylor coefficients of the determinant are polynomials of traces as in Theorem \ref{theo:traces}. This leads to 
\begin{theo}
\label{theo:existence}
The chiral Hamiltonian with $U=U_0$ and $U_0$ as in \eqref{eq:defU}, exhibits a flat band of multiplicity 2 at a real magic $\alpha_{*} \in (0.583,0.589)$,
which is minimal, in the sense that the Hamiltonian does not possess a flat band for any $\alpha$ satisfying $\vert \alpha \vert < \vert \alpha_{*} \vert$, that is,
\[     | \mathcal A \cap  (0.583,0.589) | = 1, \ \ \ \mathcal A \cap D_{\mathbb C}  ( 0 , \alpha^* ) = 
\emptyset , \]
where the counting $ | \bullet | $ respects multiplicities. In particular, the flat bands of multiplicity 2 are uniformly gapped from all other bands.
\end{theo}

\noindent
{\bf Remark.} 
Compared with results in \cite{lawa} which require floating-point arithmetic, our proof of existence relies only on exact symbolic computations, the exact evaluation of residues to compute traces of powers of $T_k$ and the summation of finitely many matrix entries to estimate the Hilbert-Schmidt norm.

\section{Preliminaries}

From now on, we consider a potential $U\in C^{\infty}(\mathbb C/\Gamma; \mathbb C)$ satisfying the first two symmetries of \eqref{eq:potential}. The last symmetry $U(\bar z)= \overline{U(z)}$ will only be needed in Corollary \ref{rat} to ensure that all traces are real. We recall that an orthonormal basis of $L^2(\mathbb C/\Gamma; \mathbb C)$ is given by setting 
$$e_{\nu}(z):=e^{\frac{i}{2} (\bar{\nu}z+\nu z)}/\sqrt{\mathrm{Vol}(\mathbb C/\Gamma)},\quad \nu \in\Gamma^*=\frac{1}{\sqrt{3}}\big(\mathbb Z+\omega\mathbb Z\big).  $$

We can express the potential $U$ in this basis.  A straightforward calculation gives 
\begin{prop}\label{prop:1}
Let $a=\frac{4\pi i}{3}(\omega a_1+\omega^2 a_2)\in\Gamma/3$, $ a_j \in \mathbb Z $. Then $u\in L^2(\mathbb C/\Gamma)$ satisfies 
\begin{equation*}
u(z+a)=\bar\omega^{(a_1+a_2)} u(z),\quad u(\omega z)={\omega} u(z),
\end{equation*}
if and only if
\begin{equation}
\label{eq:generalform}
u(z)=\sum_{n\in\mathbb Z^2} c_{n}e_{k_n}
\end{equation}
where for $n \in \ZZ^2$
$$
c_{n}=(u,e_{k_n})_{L^2(\mathbb C/\Gamma)},\quad k_n={\frac{\omega^2(2+3n_1)-\omega(2-3n_2)}{\sqrt 3}},
$$
satisfies 
\[ c_{n}=\omega c_{(-n_2)(n_1-n_2-1)}=\omega^{2} c_{(n_2-n_1+1)(-n_1)}.\] 
If in addition $\overline{u(\bar z)}=u(z)$ then
\[\overline{c_{n}}=c_{(-n_2)(-n_1)}
=\omega c_{n_1(n_1-n_2-1)}=\omega^{2} c_{(n_2-n_1-1)n_2}.\]
\end{prop}
Our aim is to obtain trace formulae for the powers of the compact operator $T_k$ defined by
in \eqref{eq:BS}.
Since odd powers of $T_k$ have only off-diagonal components, it is clear that the
traces of odd powers vanish. Thus, it is sufficient to compute the traces of powers of 
\begin{equation}
T_k^2=  \begin{pmatrix} (2D_{\bar z}-k)^{-1}U(z)(2D_{\bar z}-k)^{-1}U(-z) & 0 \\ 0 & (2D_{\bar z}-k)^{-1}U(-z)(2D_{\bar z}-k)^{-1}U(z) \end{pmatrix}. 
\end{equation}
The invariance of the trace under cyclic permutations shows that it is sufficient to compute traces of powers of 
\begin{equation}
\label{eq:Ak}
\mathcal A_k:= (2D_{\bar z}-k)^{-1}U(z)(2D_{\bar z}-k)^{-1}U(-z):L^2(\mathbb C/\Gamma; \mathbb C)\to L^2(\mathbb C/\Gamma; \mathbb C),\quad k\notin \Gamma^*.
\end{equation}

We shall study traces of powers of $\mathcal A_k$ on smaller $L^2$ spaces, which we define below, for $(p_1,p_2)\in \mathbb Z_3^2$, by
 \begin{equation}
\begin{split}
\label{eq:spaces}
L^2_{(p_1,p_2)}  ( \mathbb C/ \Gamma ; \mathbb C ) &:= \Big\{ u \in L^2 : u ( z +  2 i ( \omega a_1 + a_2 \omega^2 ) ) =
e^{ i ( a_1 p_1 + a_2 p_2 )} u ( z );\\
& \qquad \ \ a_j \in \tfrac{ 2\pi }  {3} \mathbb Z \Big\}
\end{split}
\end{equation}
whose $\CC^2$-valued analogues are defined, using \eqref{eq:4and11}, as
 \begin{equation}
\begin{split}
L^2_{(p_1,p_2)}  ( \mathbb C/ \Gamma ; \mathbb C^2 ) &:= \Big\{ u \in L^2 : \mathscr L_{a}u ( z  ) =
e^{ i ( a_1 p_1 + a_2 p_2 )} u ( z ); a_j \in \tfrac{ 2\pi }  {3} \mathbb Z \Big\}.
\end{split}
\end{equation}

We remark that the operator $(2D_{\bar z}-k)^{-1}$ acts diagonally on the Fourier basis and thus preserves the $L^2_{(p_1,p_2)}$ spaces. On the other hand, multiplication by $U(\pm z)$ does not preserve the space but one has by the translational symmetry defined in \eqref{eq:potential}
$$ U(\pm z): L^2_{(p_1,p_2)}\to L^2_{(p_1\mp1,p_2\mp1)},\quad (p_1,p_2)\in \mathbb Z_3^2. $$
In total, we have
$$  L^2_{(p_1,p_2)}\xrightarrow{U(-z)}L^2_{(p_1+1,p_2+1)}\xrightarrow{(2D_{\bar z}-k)^{-1}}
L^2_{(p_1+1,p_2+1)} \xrightarrow{U(z)}
L^2_{(p_1,p_2)}\xrightarrow{(2D_{\bar z}-k)^{-1}}L^2_{(p_1,p_2)}
. $$
This shows that we can restrict the operator $\mathcal A_k$ to the subspaces $L^2_{(p_1,p_2)}$. From now on, we will denote by 
$A_k$ the restriction of $\mathcal A_k$ to $L^2_{(1,1)}$. 
We then define the unitary multiplication operator 
\[ \begin{gathered} U_{(p_1,p_2)}:L^2_{(0,0)}(  \mathbb C/\Gamma , \mathbb C^2  ) \to L^2_{(p_1,p_2)}(  \mathbb C/\Gamma , \mathbb C^2  ) , \\ 
U_{(p_1,p_2)} v (z):=e^{\frac{i}{2}(z \bar p+ \bar z p)} v ( z ) , \ \ p = \tfrac 1 {\sqrt 3 } (\bar \omega p_1 - \omega p_2 ) , \ \ \  p_j \in \mathbb Z_3, \\  
U_{(p_1,p_2)}T_k U^*_{(p_1,p_2)} = T_{k-p}, \ \ \ k \notin \Gamma^* .
\end{gathered} \]
The $k$-independence of the spectrum of $T_k$ implies then
\begin{equation}
\label{eq:T2A0} \begin{split}  \Spec_{ L^2_{(0,0)} (  \mathbb C/\Gamma , \mathbb C^2  ) }  ( T_k^2 ) \setminus \{ 0 \} & = \Spec_{L^2_{(1,1)}
 (  \mathbb C/\Gamma , \mathbb C^2  ) } ( T_k^2 ) \setminus \{ 0 \} \\
 & =  \Spec_{L^2_{(1,1)} 
  (  \mathbb C/\Gamma , \mathbb C ) } ( A_k ) \setminus \{ 0 \} , \end{split}  
  \end{equation}
  where $ k \in D ( 0 , r ) \setminus \{ 0 \}$, and the last equality is 
  meant in the sense of sets: multiplicities of elements in the top row are twice the multiplicities of elements in the bottom row.
 \\ 
We also note that $ A_k $ is defined for $ k \in D ( 0 , r ) $ since $ D_{\bar z}^{-1} $ is
defined on $ L^2_{(p,p)} $, $ p \not \equiv 0 \!\!\! \mod 3$. 
Since $ \mathbb C \ni k \to \mathcal A_{k}|_{ L^2_{(1,1)}} $ is an analytic family of operators with compact resolvents
and the spectrum is independent of $ {k }\in \CC\setminus 3\Gamma^*$, 
 it follows  that  $\Spec(A_{k }) = \Spec(A_{ 0})$ \cite[Theorem $1.10$]{kato}.
From \eqref{eq:T2A0} we obtain, as sets,
\begin{equation}
\label{eq:T2A}  \Spec_{L^2_{(0,0)} }  ( T_p^2 ) \setminus \{ 0 \} =  \Spec_{ L^2_{(1,1)} } ( A_{k }  ) \setminus \{ 0 \},  \ \ \ p \in D( 0,r ) \setminus \{ 0 \}, \ \ k \in D ( 0 , r ) , 
\end{equation}
with multiplicities on the left, twice the multiplicities on the right. Since $k=0$ is included in the set of possible $k$ for $A_k$. Indeed, the set of possible values of $k$ is $\mathbb C\setminus \big((3\Gamma^*-i)\cup (3\Gamma^*+i)\big)$. We conclude together with \cite[Theorem $2$]{bhz2} that
\begin{equation} 
\label{eq:gapped}
\operatorname{dim}\ker_{\mathscr H }(D(\alpha)) =\operatorname{dim}\ker_{L^2_{(1,1)} }(A_0-\alpha^{-2}).
\end{equation}
We end this preliminary section by stating and proving the main three properties we will use for our calculation.
\begin{lemm}
\label{properties}
Consider a potential $U\in C^{\infty}(\mathbb C/\Gamma; \mathbb C)$ satisfying the first two symmetries of \eqref{eq:potential} with a finite number of non zero Fourier mode in its decomposition \eqref{eq:generalform}. Define the operator $ A_k$ for $k\notin (3\Gamma^*-i)\cup (3\Gamma^*+i)$, where $i:=-(\omega^2-\omega)/\sqrt 3$, to be the restriction of $\mathcal A_k$ defined in \eqref{eq:Ak} on the space $L^2_{(1,1)}$. For $\ell\geq 2$, one has:
\begin{itemize}
\item The trace is constant in $k$ 
\begin{equation}
\label{eq:constant}
\tr( A_k^{\ell})=\tau_{\ell} \text{ independent of }k\in \mathbb C \setminus (3\Gamma^*-i)\cup (3\Gamma^*+i).
\end{equation}
\item The function $\mathbb C \setminus (3\Gamma^*-i)\cup (3\Gamma^*+i) \ni k\mapsto \langle A_k^{\ell}e_{m},e_{m}\rangle_{L^2}$ is a finite sum of rational fractions on the complex plane $\mathbb C$ with degree equal to $-2\ell$ and with $($a finite number of$)$ poles contained in $(3\Gamma^*-i)\cup (3\Gamma^*+i)$.
\item For any $\gamma\in \Gamma^*$ and for any $k\notin (3\Gamma^*-i)\cup (3\Gamma^*-2i)$, we have
$$ \langle A_k^{\ell}e_{3\gamma+i},e_{3\gamma+i}\rangle_{L^2}=\langle A_{k-3\gamma}^{\ell}e_{i},e_{i}\rangle_{L^2}.$$
\end{itemize}

\end{lemm}
\begin{proof}
 The first point is a consequence of the independence of the spectrum of $T_k$ in $k$ (see \cite[\S 2.3]{bhz2}) as well as the relation \eqref{eq:T2A}.
\\
For the last two points, we prove by induction that $k\mapsto  A_k^{\ell}e_{3\gamma+i}$, where $\gamma\in \Gamma^*$, is of the form 
\begin{equation}
\label{eq:Induc}
A_k^{\ell}e_{3\gamma+i}=\sum_{\nu\in F}R_{\nu+3\gamma}(k)e_{\nu+3\gamma}, 
\end{equation}
where $F\subset (3\Gamma^*+i)$ is a finite set and $R_{\nu}(k)$ is a sum of rational fraction of degree $-2\ell$ with poles located on $(3\Gamma^*-i)\cup (3\Gamma^*+i)$. Moreover, we will prove that the one has the relation $R_{\nu+3\gamma}(k)=R_{\nu}(k-3\gamma).$ 
\\
The result is clear for $\ell=0$. Suppose the result true for $\ell$, let's prove it holds for $\ell+1$.
The main observation is that multiplication by $U(\pm z)$ acts as a shift on the Fourier basis. The multiplication by $U(-z)$ sends $e_{\nu}$ to a linear combination of $e_{\ell}$ for $\ell \in (3\Gamma^*+2i)$. Then applying $(D(0)-k)^{-1}$ multiplies the coefficient of $e_{\ell}$ by $(\ell-k)^{-1}$. Multiplying by $U(z)$ gives back a linear combination of $e_{\nu}$ with $\nu\in (3\Gamma^*+i)$. Finally, applying $(D(0)-k)^{-1}$ multiplies the coefficient of $e_{\nu}$ by $(\nu-k)^{-1}$. This means that, using the induction hypothesis \eqref{eq:Induc},
$$ A_k^{\ell+1}e_{3\gamma+i}=\sum_{\nu\in F}R_{\nu+3\gamma}(k)\sum_{\eta\in L}\sum_{\beta\in L}\frac{a_{\eta}}{k-(\nu+\beta-\eta+3\gamma)}\frac{a_{\beta}}{k-(\nu+\beta+3\gamma)}e_{\nu+\beta-\eta+3\gamma},$$
where $L\subset 3\Gamma^*+i$ is a finite subset that depends only on $U$ and $a_{\bullet}$ are constants. Thus, it is clear from this formula that the induction carries on to $\ell+1$. This concludes the proof of the Lemma.
 \end{proof}

\section{Trace computations}
We prove the following result.

\begin{theo}
\label{traceresult}
Let  $A_k:L^2_{(1,1)}\to L^2_{(1,1)}$ be a meromorphic family of Hilbert-Schmidt operators defined for $k\notin (3\Gamma^*-i)\cup (3\Gamma^*+i)$. We suppose that $A_k$ satisfies the three properties stated in Lemma \ref{properties}. Then one has, for any $\ell\geq 2$,
\begin{equation}
\label{eq:taul}
\begin{split}
\tau_{\ell} & =\frac{2i\pi\omega}{3\sqrt 3}\sum_{n\in \mathbb Z} n\left[\sum_{m\in \mathbb Z}\mathrm{Res}\big(\langle A_k^{\ell}e_{i},e_{i}\rangle_{L^2},\sqrt 3(m\omega^2-n\omega)+i\big)\right.
\\  & \ \ \ \ \ \ \ \ \ \ \ \ \left.+\sum_{m\in \mathbb Z}\mathrm{Res}\big(\langle A_k^{\ell}e_{i},e_{i}\rangle_{L^2},\sqrt 3(m\omega^2-n\omega)+2i\big)\right], \end{split} \end{equation} 
where the all infinite sums are in fact finite.
\end{theo}

\begin{proof}
We want to give a semi-explicit formula for $\tau_{\ell}$ in terms of the residue of the rational fraction (second point in Lemma \ref{properties})
$$k\in \mathbb C\setminus \big(3\Gamma^*-i\big)\cup\big(3\Gamma^*+i\big)\mapsto   \langle A_k^{\ell}e_{i},e_{i}\rangle_{L^2}.$$
We first start by writing, using that $A_k^{\ell}$ is trace-class for $\ell\geq 2$,  
$$\tau_{\ell}=\sum_{\gamma\in \Gamma} \langle A_k^{\ell}e_{3\gamma+i},e_{3\gamma+i}\rangle_{L^2}. $$
We start with the relation which follows directly from \eqref{eq:constant},
$$3\tau_{\ell}=\int_0^3\tr\big(A^{\ell}_{t\omega^2/\sqrt 3}\big)dt=\int_0^3\sum_{n\in \mathbb Z^2}\langle A_{t\omega^2/\sqrt 3}^{\ell}e_{3\gamma_{n}+i},e_{3\gamma_{n}+i}\rangle_{L^2}dt. $$
Here, we wrote $\gamma_{n}:=\frac{1}{\sqrt 3}(n_1\omega^2-n_2\omega)\in \Gamma^*$.
We now use the third property stated in Lemma \ref{properties} to write
$$3\tau_{\ell}=\int_0^3\sum_{n\in \mathbb Z^2} \langle A_{t\omega^2/\sqrt 3-3\gamma_{n}}^{\ell}e_{i},e_{i}\rangle_{L^2}dt. $$
The second property  in  Lemma \ref{properties}  implies that 
\begin{equation}
\label{eq:decA} \langle A_k^{\ell}e_{i},e_{i}\rangle_{L^2} =  {\mathcal O}\big(k^{-2\ell}\big). 
\end{equation}
Since we assume that $\ell\geq 2$, this justifies the exchange of integration and summation such that $$3\tau_{\ell}=\sum_{n\in\mathbb Z^2} \int_0^3\langle A^{\ell}_{(t+3n_1)\omega^2/\sqrt 3-3n_2 \omega/\sqrt 3}e_{i},e_{i}\rangle_{L^2}dt.$$
We make the change of variable $s=t+3n_1$ and sum in $n_1$ to get
\begin{equation}
\label{eq:int}
3\tau_{\ell}=\sum_{n\in\mathbb Z}\int_{\mathbb R}\langle A^{\ell}_{t\omega^2/\sqrt 3-3n \omega/\sqrt 3}e_{i},e_{i}\rangle_{L^2}dt. 
\end{equation}
We now consider 
$$\int_{\mathbb R}\langle A^{\ell}_{t\omega^2/\sqrt 3-3n \omega/\sqrt 3}e_{i},e_{i}\rangle_{L^2} \ dt-\int_{\mathbb R}\langle A^{\ell}_{t\omega^2/\sqrt 3-3(n+1) \omega/\sqrt 3}e_{i},e_{i}\rangle_{L^2}dt, \ \ 
n \in \mathbb Z.  $$
This is equal to the limit of the integral over a parallelogram $\Gamma_{n,R}$ with 
sides $\frac{1}{\sqrt 3}[-R\omega^2-3n\omega,R\omega^2-3n\omega]$, $\frac{1}{\sqrt 3}[R\omega^2-3n \omega,R\omega^2-3(n+1) \omega]$, $\frac{1}{\sqrt 3}[R\omega^2-3(n+1) \omega,-R\omega^2-3(n+1) \omega]$ and $\frac{1}{\sqrt 3}[-R\omega^2-3(n+1) \omega,-R\omega^2-3n \omega]$. Here, we used \eqref{eq:decA}  to prove that the integral over the small parallel sides tends to $0$. In particular, because there is only a finite number of poles, we see that, for $|n|$ large enough, one has
$$\int_{\mathbb R}\langle A^{\ell}_{t\omega^2/\sqrt 3-3n \omega/\sqrt 3}e_{i},e_{i}\rangle_{L^2}dt-\int_{\mathbb R}\langle A^{\ell}_{t\omega^2/\sqrt 3-3(n+1) \omega/\sqrt 3}e_{i},e_{i}\rangle_{L^2}dt=0. $$
Using formula \eqref{eq:int} as well as a partial summation, this allows us to rewrite the full trace as a telescopic sum 
\begin{equation}
\label{eq:telescope} 3\tau_{\ell}=\sum_{n\in\mathbb Z}n\left[ \int_{\mathbb R}\langle A^{\ell}_{t\omega^2/\sqrt 3-3n \omega/\sqrt 3}e_{i},e_{i}\rangle_{L^2} \ dt-\int_{\mathbb R}\langle A^{\ell}_{t\omega^2/\sqrt 3-3(n+1) \omega/\sqrt 3}e_{i},e_{i}\rangle_{L^2}dt\right]. \end{equation}
The residue theorem shows that for $n\in \mathbb Z$ and $R$ large enough,
\begin{equation}
\label{eq:omega}  \begin{split} \int_{\Gamma_{n,R}}\langle A^{\ell}_{z}e_{i},e_{i}\rangle_{L^2}dz & =\int_{\mathbb R}\langle A^{\ell}_{t\omega^2/\sqrt 3-3n \omega/\sqrt 3}e_{i},e_{i}\rangle_{L^2}\frac{\omega^2}{\sqrt 3}dt \\
& \ \ \ \ \ - \int_{\mathbb R}\langle A^{\ell}_{t\omega^2/\sqrt 3-3(n+1) \omega/\sqrt 3}e_{i},e_{i}\rangle_{L^2})\frac{\omega^2}{\sqrt 3}dt. \end{split} \end{equation}
Applying the residue theorem and using \eqref{eq:telescope} gives \eqref{eq:taul}.
\end{proof}

The consequences of this formula are summarized in the following theorem:
\begin{theo}
\label{rat}
Consider a potential $U\in C^{\infty}(\mathbb C/\Gamma; \mathbb C)$ satisfying the first two symmetries of \eqref{eq:potential} with finitely many non-zero Fourier modes $c_{n} \in \mathbb Q(\omega/\sqrt 3)$ appearing in the decomposition \eqref{eq:generalform}. Then for any $\ell \geq 2$, one has $\tau_{\ell}\in \pi\mathbb Q(\omega/\sqrt 3) $. If $U$ also has the third symmetry of \eqref{eq:potential} then the traces are real and thus $\tau_{\ell}\in  \pi \mathbb Q/\sqrt 3$. In particular, for all potentials satisfying all three symmetries in \eqref{eq:potential}, including $U = U_0$ defined in \eqref{eq:defU}, one has
$$\forall \ell\geq 2,\quad \tr(T_k^{2\ell})=\sum_{\alpha\in \mathcal A(U)}\alpha^{-2\ell}=\frac{\pi}{\sqrt 3}q_{\ell},\;\;\; q_{\ell}\in \mathbb Q, $$
where $\mathcal A(U)$ is the set of magic angles counting multiplicity for a potential $U.$
\end{theo}
\begin{proof}
Under the hypothesis of the corollary, the function $k\mapsto \langle A_k^{\ell}e_{i},e_{i}\rangle_{L^2}$ is a rational fraction with coefficients in $\mathbb Q(\omega/\sqrt 3)$. This ring is actually a field as $\omega/\sqrt 3$ is algebraic on $\mathbb Q$. Now, taking partial fraction expansion of $k\mapsto \langle A_k^{\ell}e_{i},e_{i}\rangle_{L^2}$ in $\mathbb Q(\omega/\sqrt 3)(X)$ (the space of rational fractions with coefficients living in $\mathbb Q(\omega/\sqrt 3)(X)$) gives coefficients in $\mathbb Q(\omega/\sqrt 3)$. In particular, the residues of $k\mapsto \langle A_k^{\ell}e_{i},e_{i}\rangle_{L^2}$ live in the field $\mathbb Q(\omega/\sqrt 3)(X)$. But the uniqueness of the partial fraction expansion now gives that these are also the residues of $k\mapsto \langle A_k^{\ell}e_{i},e_{i}\rangle_{L^2}$ in $\mathbb C(X)$. Using the trace formula stated in Theorem \ref{traceresult}, this yields 
$$\forall \ell\geq 2,\quad \tau_{\ell}\in  \pi \mathbb Q(\omega/\sqrt 3). $$
If we add the last symmetry of \eqref{eq:potential}, the trace is real so that
$$ \forall \ell\geq 2,\quad \tau_{\ell}\in \pi\mathbb Q(\omega/\sqrt 3)\cap \mathbb R\Rightarrow \tau_{\ell}=\frac{\pi}{\sqrt 3}q_{\ell},\;\;\; q_{\ell}\in \mathbb Q. $$
\end{proof}
This rationality condition suffices to prove that there is an infinite number of magic angles as long as there exists at least one magic angle.
\begin{theo}
\label{sec:races_to_infinity}
Under the assumptions and with the same notation as in Theorem \ref{rat} one has the implication
$$|\mathcal A(U)|>0\Rightarrow |\mathcal A(U)|=+\infty.$$
In particular, the set of magic angles for our canonical potential $U_0$ defined in \eqref{eq:defU} is infinite.
\\
Let $N\geq 0$, for a tuple $a=(a_{n})_{\{n; \Vert n \Vert_{\infty} \leq N\}}$, define $U_a$ to be the potential defined by \eqref{eq:generalform}. Then the above implication holds for a generic $($in the sense of Baire$)$ set of coefficients $a=(a_{n})_{\{n;\Vert n \Vert_{\infty} \le N\}}\in \mathbb C^{(2N+1)^2}$ that contains $(\mathbb Q(\omega/\sqrt{3}))^{(2N+1)^2}.$
\end{theo}
\begin{proof}
We start by observing that since $\pi$ is transcendental on $\mathbb Q$, it is also transcendental in $\mathbb Q(\omega/\sqrt{3}).$
Now, assume by contradiction, that there exist only finitely many eigenvalues $\lambda_i \in \mathbb C$ for $i=1,..,N$ of $A_{k}^2$. Then we define the $n$-th symmetric polynomial
\[ e_n (\lambda_1 , \ldots , \lambda_N )=\sum_{1\le  j_1 < j_2 < \cdots < j_n \le N} \lambda_{j_1} \dotsm \lambda_{j_n}.\]
Newton identities show that this polynomial can be expressed as 
\begin{equation}
\label{eq:Newton} e_n(\lambda_1 , \ldots , \lambda_N ) = (-1)^n  \sum_{m_1 + 2m_2 + \cdots + nm_n = n \atop m_1 \ge 0, \ldots, m_n \ge 0} \prod_{i=1}^n \frac{(-\tr A_{k}^{2i})^{m_i}}{m_i ! i^{m_i}}\end{equation}
where $e_n=0$ for $n>N.$
Theorem \ref{theo:traces} shows that
 $$\prod_{i=1}^n (\tr A_{k}^{{ 2i}})^{m_i} \in \mathbb Q \left(\frac{\omega}{\sqrt{3}}\right)\pi^{m_1 \cdots m_n}.$$
The power $m_1\cdots m_n$ from sequences allowed in \eqref{eq:Newton} is maximized by the unique choice $m=(n,0,\hdots,0)$.  The Newton identities for $n>N$ then imply that the transcendental number $\pi$ is a root of a polynomial with coefficients in $\mathbb Q \left({\omega}/{\sqrt{3}}\right).$ But then all these coefficients vanish, this is equivalent to the fact that the spectrum is empty (because of the determinant function, see \eqref{eq:Fredholm_det}). For our particular choice of potential $U_0$, the fact that $ \tr A^2_{k}=0$ contradicts \cite[Theorem $3$]{beta} so the set of magic angles is non empty, and thus infinite.
\\
Now, let $a=(a_{n})_{\{n;\Vert n \Vert_{\infty} \le N\}} \in \mathbb C^{(2N+1)^2}$ and assume that $\mathcal A(U_a) \neq \emptyset$. Then, we can find an open neighbourhood of $ a $,  $\Omega_a \ni a$, such that for coefficients $b=(b_{n})_{\{n;\Vert n \Vert_{\infty} \le N\}} \in \Omega_a$ we have $\mathcal A(U_b) \neq \emptyset$. 
Take  $q=(q_{n})_{\{n;\Vert n \Vert_{\infty} \le N\}} \in (\mathbb Q(\omega/\sqrt{3}))^{(2N+1)^2}\cap \Omega_a$ 
for which we then have $\vert \mathcal A(U_q) \vert = \infty.$ 
Continuity of eigenvalues of $T_k $ as the potential $ U $ changes shows that the 
$ V_{m,a} := \{ b \in \Omega_a :  | \mathcal A(U_{b}) | \geq m \}  $
is open and dense in $\Omega_a$. Hence, the set 
coefficients for which $ 0 < | \mathcal A_b | < \infty $ is given by $\bigcup_{m \in \mathbb N} \bigcup_{q \in  ( \mathbb Q+i\mathbb Q)^{2N+1}} \Omega_q \setminus V_{m,q}$. It is then meagre and does not contain $(\mathbb Q(\omega/\sqrt{3}))^{(2N+1)^2}.$
\end{proof}

\section{Fredholm determinants and the first magic angle}
In this section, we explain how to compute the first few traces from our formula and show the existence of a simple real magic angle, i.e. prove Theorem \ref{theo:existence}.
From now on, our choice of potential is given by $U=U_0$ defined in \ref{eq:defU}.
Here we recall some facts from \cite{beta,bhz2} needed in this paper.
\subsection{Fourier coordinates}
For our numerics, it is convenient to use rectangular coordinates $z=2i(\omega y_1+\omega y_2)$, see \cite[\S 3.3]{beta} for details. In these coordinates, we may introduce
\begin{equation}
\label{eq:Dkal} \begin{gathered}    
\mathcal D_{k } :={ \omega^2  
(D_{y_1} + k_1) - \omega (D_{y_2} + k_2 ) }, \\
\mathcal V ( y ) := \sqrt 3 ( e^{ -i ( y_1 + y_2 ) } 
+ \omega e^{ i (  2 y_1 - y_2 ) } + \omega^2 e^{ i ( - y_1 + 2 y_2 )}  ),
\end{gathered} 
\end{equation}
with {\em periodic} periodic boundary conditions (for $ y \mapsto
y + 2 \pi \mathbf n  $, $ \mathbf n  \in  \ZZ^2 $). {In the following, we shall write $\mathcal V_{\pm}(y):=\mathcal V(\pm y).$}
The operator $\mathcal A_k$, defined in \eqref{eq:Ak}, reads in the new coordinates
\[ \begin{gathered} 
\mathcal D_k^{-1} \mathcal V_+ \mathcal D_{k}^{-1} \mathcal V_- :
L^2 (\mathbb C/ 2 \pi (\mathbb Z+i\mathbb Z)  ; \mathbb C ) \to L^2 ( \mathbb C/2 \pi (\mathbb Z+i\mathbb Z) ; \mathbb C ).\end{gathered} \]
On the Fourier transform side
we introduce the equivalent of operators \eqref{eq:Dkal}
\begin{equation} \begin{gathered}   \widehat{ \mathcal D}_{k } :={ \omega^2  
(D + k_1) - \omega (D + k_2 ) }, \text{ with } D= \operatorname{diag}(\ell)_{\ell \in \ZZ} \\
\widehat{\mathcal V}_{\pm} ( y ) :=\sqrt 3\left(  J^{\pm} \otimes J^{\pm} 
+ \omega J^{\mp 2}\otimes J^{\pm} + \omega^2 J^{\pm} \otimes J^{\mp 2}\right),
\end{gathered} 
\end{equation}
where $J$ is the right-shift $J((a_n)_n) = (a_{n+1})_n $ -- see 
\cite[(3.17)]{beta}. The spaces  
$ L^2_{(p_1,p_2)}  ( \mathbb C/\Gamma ; \mathbb C ) $, introduced in \eqref{eq:spaces}, correspond to 
\[ 
\begin{gathered} \ell_{(p_1 ,p_2)}^2:=\{f\in \ell^2(\mathbb{Z}^2) : \forall \, n \notin \left(3\mathbb{Z}+p_1\right)\times\left(3\mathbb{Z}+p_2\right),\:\: f_{n}=0\}.
 \end{gathered} \]
As in \cite[\S 3.3]{beta}, we introduce auxiliary operators
$J^{p,q}:=J^p\otimes J^q,\:\: p,q\in \mathbb{Z}$. 
For a diagonal matrix $\Lambda=(\Lambda_{i,j})_{i,j\in \mathbb{Z}}$ acting on $\ell^2(\mathbb{Z}^2)$, we define a new diagonal matrix
$$\Lambda_{p,q}:=(\Lambda_{i+p,j+q})_{i,j\in \mathbb{Z}} .$$
We recall the following properties \cite[(3.24)]{beta}
\begin{equation}
\label{eq:identity}
J^{p,q}\Lambda J^{p',q'}=\Lambda_{p,q}J^{p+p',q+q'}=J^{p+p',q+q'}\Lambda_{-p',-q'}. 
\end{equation}
Denoting the inverse of $\widehat{ \mathcal D}_k^{-1}$ by
$$\Lambda=\Lambda_k:=\widehat{ \mathcal D}_k^{-1}, \ \ \  \Lambda_{m,n}=\frac{1}{\omega^2(m+k_1)-\omega(n+k_2)}, \ \ \  (k_1,k_2)\notin \mathbb{Z}^2, $$
we see that $\mathcal A_k$ reads in the new Fourier coordinates
$$\frac 1 3\widehat{\mathcal{A}}_k=\Lambda\Lambda_{1,1}+\omega\Lambda\Lambda_{1,-2}+\omega^2\Lambda\Lambda_{-2,1}+\omega\Lambda\Lambda_{1,1}J^{3,0}+\omega^2\Lambda\Lambda_{1,1}J^{0,3} $$
$$+\omega\Lambda\Lambda_{-2,1}J^{-3,0}+\omega^2\Lambda\Lambda_{1,-2}J^{0,-3}+\Lambda\Lambda_{-2,1}J^{-3,3}+\Lambda\Lambda_{1,-2}J^{3,-3} $$
with $\widehat{A}_k$ the analogous restriction to $\ell^2_{(1,1)}.$

\subsection{Fredholm determinants}
We start by defining the regularized Fredholm determinant
\begin{equation}
\label{eq:Fredholm_det}
\operatorname{det}_2(1-\alpha^2 \widehat A_{k}) = \prod_{\lambda \in \Spec(A_{k})} E_1(\alpha^2 \lambda) \text{ with }E_1(z) = (1-z)e^z
\end{equation}
where the product respects multiplicities.
We find from \eqref{eq:T2A0} that $\operatorname{det}_2(1-\alpha^2 \widehat{A}_{k})  =0 \Leftrightarrow \alpha^{-1} \in \Spec(T_{k})\setminus\{0\}.$
The symmetry of the spectrum of $\widehat{A}_{k}$, $\Spec(\widehat{A}_{k})= \overline{\Spec(\widehat{A}_{k})},$ implies that $\alpha \mapsto \operatorname{det}_2(1-\alpha^2 \widehat{A}_{k}) $ is real-valued on the real axis.
To show existence and simplicity of magic angles, in the representation, we therefore use the following Lemma which provides ab initio bounds on the Fredholm determinants and its derivatives.

\begin{lemm}
\label{lemm:lemma}
The determinant $\CC \ni \alpha \mapsto \operatorname{det}_2(1-\alpha^2\widehat{A}_{k})$ in \eqref{eq:Fredholm_det} is an entire function, independent of $k \in \CC$, 
which for any $n,m \in \mathbb N_0$ satisfies
\[ \begin{split}
\left\lvert \partial_{\alpha}^m \operatorname{det}_2(1-\alpha^2 \widehat{A}_{k}) -\partial_{\alpha}^m\sum_{j=0}^{n}  \mu_j \frac{(-\alpha^{2})^j}{j!} \right\rvert &\le \sum_{j=n+1}^{\infty} \partial_{\vert \alpha \vert}^m \Bigg( \frac{\sqrt{e}\inf_{k \in \CC}\Vert \widehat{A}_{k} \Vert_2  \vert \alpha\vert^{2}}{\sqrt{j}} \Bigg)^j\end{split} \] 
with $\Vert A_{0} \Vert_2 \le 2$, where 
\begin{equation}
\label{eq:muk}
 {\mu}_j :=  \operatorname{det}\begin{pmatrix} 0 & j-1 &0 & \cdots & 0 \\
\sigma_2 & 0 & j-2 & \cdots & 0 \\
\vdots & \vdots & \ddots & \ddots & \vdots \\ 
\sigma_{j-1}& \sigma_{j-2} & \cdots  & 0  & 1\\
\sigma_{j} & \sigma_{j-1}  & \sigma_{j-2}& \cdots & 0\end{pmatrix}, \text{ with }  \sigma_j= \tr \widehat{A}_{k}^{j}.
\end{equation}
\end{lemm}
\begin{proof}
The expression \eqref{eq:Fredholm_det} is well-defined since $\widehat{A}_{k}$ is a Hilbert-Schmidt operator and the Taylor coefficients $\mu_j$ are for example stated in \cite[(6.13)]{Simon}.
Indeed, since $\vert E_1 \vert \le e^{\frac{\vert z \vert^2}{2}}$ for $z \in \CC$ and $\sum_{\lambda \in \Spec(\widehat{A}_{k})} \vert \lambda \vert^2 \le \Vert \widehat{A}_{k} \Vert^2_2,$
we conclude that
\begin{equation}
\label{eq:growth_bd}
 \vert \operatorname{det}_2(1-\alpha^2\widehat{A}_{k}) \vert \le \exp \left( {\frac{\vert \alpha \vert^4 \Vert \widehat{A}_{k} \Vert^2_2}{2}} \right) .
 \end{equation}
Cauchy estimates for the entire function $f(z):=\det_2(1+z \hat{A}_k)$ show by using the growth bound \eqref{eq:growth_bd}
$$\vert \mu_j \vert \le \frac{j!}{\vert \alpha \vert^{2j}}
\exp \left( {\frac{\vert \alpha \vert^4 \Vert \widehat{A}_{k} \Vert^2_2}{2}} \right) 
$$ 
which is optimized at $\vert \alpha \vert^2 = \frac{\sqrt{j}}{ \Vert \widehat{A}_{k} \Vert_2}$, such that $$\vert \mu_j \vert \le \frac{\Vert \widehat{A}_{k} \Vert^j_2e^{j/2}j! }{j^{j/2}}.$$
The Taylor coefficients $\mu_j$ are then given by the Plemelj-Smithies formula \cite{Simon} stated in \eqref{eq:muk}. Since they only depend on traces $\sigma_j$ which are independent of $k$, it follows that the regularized Fredholm determinant is an entire function independent of $k.$ Hence, it suffices to study the determinant for $k=0.$

If we write $\widehat{A}_{0}=(\widehat{A}_0(n))_{n \in \mathbb Z^2}$ and let $P_m$ be the projection onto $(3\{-m,-m+1,...,m\}+1)^2,$ then 
\begin{equation}
 \begin{split}
\label{eq:estm0}
\Vert  \widehat{A}_{0} \Vert_2&\le \Vert  P_M \widehat{A}_{0}\Vert_2+\Vert (\operatorname{id}-P_M)\widehat{A}_{0}  \Vert_2.
\end{split}
\end{equation}

The first term constitutes the Hilbert-Schmidt norm of a finite matrix which can be explicitly computed from the matrix elements using symbolic calculations, indeed \[\Vert  P_M \widehat{A}_{0}\Vert_2 = \sqrt{\tr(P_M \widehat{A}_{0}\widehat{A}_{0}^* P_M)}  \le 5\text{ for }M=760.\]

To estimate the second term, we may use that the operator norm of $ \mathcal {V}_{\pm} $ satisfies $\Vert \mathcal {V}_{\pm} \Vert = 3\sqrt{3}$, therefore one has
\begin{equation}
\begin{split}
\label{eq:HS}
\Vert (\operatorname{id}-P_M)\widehat{A}_{0} \Vert_2 &\le 9 \Vert 
  (\operatorname{id}-P_M) (\mathscr D_{0}^{-1})_{\ell^2_{(1,1)} \to \ell^2_{(1,1)}} \Vert_4  \Vert (\mathscr D_0^{-1})_{\ell^2_{(2,2)} \to \ell^2_{(2,2)}} \Vert_4.
  \end{split}
\end{equation}
We recall that by definition 
\[ \Vert (\mathscr D_{0}^{-1})_{\ell^2_{(1,1)} \to \ell^2_{(1,1)}}\Vert_4=\Bigg( \sum_{m \in (3\ZZ+1)^2} \vert \omega^2 m_1-\omega m_2 \vert^{-4} \Bigg)^{1/4}.\]
A simple change of variables shows that $\Vert (\mathscr D_{0}^{-1})_{\ell^2_{(1,1)} \to \ell^2_{(1,1)}}\Vert_4 =  \Vert (\mathscr D_{0}^{-1})_{\ell^2_{(2,2)} \to \ell^2_{(2,2)}}\Vert_4.$
Then, a direct computation shows that in terms of $$g(m) = \frac{3((m_1+1)^2+(m_2+1)^2+(m_1+m_2)^2)}{2}-2$$ we have
\[ \begin{split}
\Vert (\mathscr D_{0}^{-1})_{\ell^2_{(2,2)} \to \ell^2_{(2,2)}} \Vert_4
= \frac{1}{\sqrt{3}}\Bigg(\sum_{m \in \ZZ^2} \frac{1}{g(m)^2} \Bigg)^{1/4}.
\end{split}\]
While an explicit computation shows using exact symbolic calculations
\begin{equation}
\label{eq:main_part}\sum_{ \vert m \vert_{\infty} \le 6} \frac{1}{g(m)^2} \le \frac{24}{7}
\end{equation} Then, we may use for $\vert m \vert_{\infty}> 6$ that  $g(m)  \ge \vert m \vert^2+5^2$, such that
we can estimate the remainder
\begin{equation}
\begin{split}
\label{eq:remainder}
&\sum_{ \vert m \vert_{\infty} \ge 7} \frac{1}{g(m)^2} \le \int_6^{\infty} \frac{2\pi r}{(r^2+5^2)^2} \ dr  = \frac{\pi}{61}\\
&\Rightarrow \Vert (\mathscr D_{0}^{-1})_{\ell^2_{(2,2)} \to \ell^2_{(2,2)}}  \Vert_4 \le  \left(\frac{8}{21}+\frac{\pi}{549}\right)^{1/4}.
\end{split}
\end{equation}
Inserting this estimate into \eqref{eq:HS}, we find along the lines of \eqref{eq:remainder}
\begin{equation}
\begin{split}
\label{eq:estm2}
\Vert (\operatorname{id}-P_M)\widehat{A}_{0} \Vert_2 &\le \frac{213}{10} \Vert 
  (\operatorname{id}-P_M) (\mathscr D_{0}^{-1})_{\ell^2_{(1,1)} \to \ell^2_{(1,1)}} \Vert_4 \\
  &\le \frac{213}{10} \frac{1}{\sqrt{3}} \Bigg(\int_{760}^{\infty} \frac{2\pi r}{(r^2+759^2)^2} \ dr \Bigg)^{1/4}<\frac{1}{2},
  \end{split} 
\end{equation}
which shows that $\Vert \widehat{A}_0 \Vert_2 <\frac{11}{2}.$
\end{proof}
Using the preceding error estimate with the explicit traces in Table \ref{table:traces}, we conclude the existence of a first real magic angle in the next Proposition. The Proposition also completes the proof of Theorem \ref{theo:existence}. Indeed, \eqref{eq:gapped} implies together with \cite[Theorem 2]{bhz2} the existence of a 0 gap between the two flat bands of the Hamiltonian and the remaining bands.
\begin{prop}
There exists a simple real eigenvalue $\frac{1}{\alpha_{*}^2}$ to the operator $\widehat{A}_{k}$, independent of $k \in \CC$, with $\alpha_{*} \in (0.583,0.589)$ such that $(\frac{1}{\alpha_{*}^2},\infty) \subset \RR \setminus \Spec(\widehat{A}_{k}).$
\end{prop}
\begin{proof}
To see that this is the first real magic angle, we first notice that the operator norm of $\widehat{A}_{0}$ is bounded by
\[ \Vert \widehat{A}_{0} \Vert \le (3\sqrt{3})^2  \Vert (\mathscr D_{0}^{-1})_{\ell^2_{(1,1)} \to \ell^2_{(1,1)}} \Vert^2 = 9.\]
This estimate shows that $\alpha \in \mathbb R^+$ with $ 1/(\alpha^2) \in \Spec(\widehat{A}_{0}) $ satisfies $\alpha \ge \frac{1}{3}.$
A finite number of traces as explicitly computed in Table \ref{table:traces} are then relevant to prove the existence of a magic angle. 
\begin{table}
\begin{subtable}{.5\linewidth}
  \centering
\begin{tabular}{P{1cm} |P{2.5cm} P{2.5cm}|}
    $p$ & $\sigma_p\frac{\sqrt{3}}{\pi}$ \\[1ex]
    \hline\hline
    {\color{blue}1} & {\color{blue}$2/3$} \\[1ex]
    \hline
    2 & $4$\\[1ex]
    \hline
    3 &$96/7\approx \textit{13.71} $ \\[1ex]
    \hline
     4 & $40$ \\[1ex]
    \hline
\end{tabular}
\end{subtable}%
\begin{subtable}{.5\linewidth}
  \centering
\begin{tabular}{P{1cm} |P{6.5cm} P{5.5cm}|}
    $p$ & $\sigma_p \frac{\sqrt{3}}{\pi}$ \\[1ex] 
    \hline\hline
     5 & ${28680}/{247} \approx \textit{116.14}$ \\[1ex]
    \hline
    6 & $ 2206080/6517 \approx \textit{338.51}$ \\[1ex]
    \hline
    7 & $ 1957475168/1983163 \approx \textit{987.05}$ \\[1ex]
    \hline
    8 & $ 39948260880/13882141 \approx \textit{2877.67}$ \\[1ex]
    \hline
  \end{tabular}
\end{subtable}
\caption{First eight exact traces of $A^{p}_{k}$, $\sigma_p = \tr(A_k^p)$, with floating point approximation, where ${\color{blue}\sigma_1}:=\lim_{n \to \infty} \sum_{\vert i \vert \le n} \langle A_k e_i,e_i \rangle$ is not absolutely summable as $A_{k}$ is not of trace-class, computed using Theorem \ref{traceresult} in the version stated as Theorem \ref{theo:technical_tristan} in the appendix. One sees that the ratio of $\sigma_p/\sigma_{p-1}\approx 1/0.5857^2 = 2.91507,$ for $p$ large, where $0.5857$ is the first magic angle.}
\label{table:traces}
\end{table}

For $\nu \in \mathbb R^+$ we find 
$$ r_i \le \bigg(\frac{2\nu}{\alpha}\bigg)^i \frac{\big(\frac{\nu}{\sqrt{N}}\big)^{N-i}}{1-\frac{\nu}{\sqrt{N}}} \text{ for }
 r_0:=\sum_{k=N}^{\infty} \bigg(\frac{\nu}{\sqrt{k}} \bigg)^k \text{ and } r_1:=\sum_{k=N}^{\infty} \frac{2k}{\alpha} \bigg(\frac{\nu}{\sqrt{k}} \bigg)^k.$$
Evaluating the bound for $N=17$ and $\nu=\sqrt{e} \Vert \widehat{A}_{0} \Vert_2 \alpha^2$, as in the error bound stated in Lemma \ref{lemm:lemma}, with upper bound $\Vert \widehat{A}_{0} \Vert_2 = 5.5$, we obtain for $\alpha=\frac{3}{5}$ that $r_0 \le \frac{1}{50}$ and $r_1 \le \frac{1}{2}$.
The existence of a root follows from studying 
$$ f(\alpha) : =\sum_{k=0}^{16}  \mu_k \frac{(-\alpha^2)^{k}}{k!}, \ \ \ 
\sup_{\alpha \in (1/3,\beta)} f'(\alpha) \leq g(\beta=0.6) := \sum_{k=2}^{20}a_k(\beta),$$
where the summation starts at $k=2$ since $\mu_1=0$, with
\[  a_k(\beta) = \begin{cases}  2\mu_k \frac{(-1)^k\big(\tfrac{1}{3}\big)^{2k-1}}{(k-1)!},&\quad \text{ if } \mu_k (-1)^k<0 \\
2\mu_k \frac{(-1)^k(\frac{3}{5})^{2k-1}}{(k-1)!}, &\quad \text{ if } \mu_k (-1)^k \ge 0. \end{cases}\]
One then checks (using computations involving integers only)
\[ \begin{split}  & f(0.583)>\frac{1}{40}, \qquad f(0.589)<-\frac{1}{40}, \text{ and } g(\frac{3}{5})<-\frac{7}{10}.\end{split} \]

We conclude that there is $\alpha_* \in (0.583,0.589)$ such that $\operatorname{det}_2(1-\alpha_*^{2}\widehat{A}_{k})=0$ and $\partial_{\alpha} \vert_{\alpha=\alpha_{*}}\operatorname{det}_2(1-\alpha^2\widehat{A}_{k}) <0.$ The non-existence of any other $\alpha \in (\tfrac{1}{3},\alpha_{*})$ at which the determinant vanishes follows from the monotonicity of $f.$
\end{proof}

\vspace{0.0cm}
\begin{center}
\noindent
{\sc  Appendix: Trace formula in Fourier coordinates.}
\end{center}
\renewcommand{\theequation}{A.\arabic{equation}}
\refstepcounter{section}
\renewcommand{\thesection}{A}
\setcounter{equation}{0}

In this section we give an auxiliary version of Theorem \ref{traceresult} that we used for our computer assisted computation of traces.
Using the relation \eqref{eq:identity}, the diagonal part of $\mathcal{A}_k^{\ell}$ is of the form
\begin{equation}
\label{eq:form}
\begin{gathered} 
((\mathcal{A}_k^{\ell})_{ii})_{i \in \ZZ}=3^{\ell}\sum_{\pi\in \Theta_{\ell}}\omega^{m_{\pi}}\prod_{i=1}^{\ell}\Lambda_{\tilde{\alpha}_i,\tilde{\beta}_i}\Lambda_{{\tilde{\gamma}_i},{\tilde{\delta}_i}},\\\pi: =\left[(\alpha_1,\beta_1),(\gamma_1,\delta_1),(\alpha_2,\beta_2),...,(\gamma_{\ell},\delta_{\ell})\right], 
\end{gathered}
\end{equation}
where
\begin{equation}
\label{eq:coefficients} \begin{gathered} \tilde{\alpha}_i=\sum_{j=1}^{i-1}\alpha_j+\gamma_j \quad \tilde{\beta}_i=\sum_{j=1}^{i-1}\beta_j+\delta_j, \quad {\tilde{\gamma}_i}=\alpha_i+\sum_{j=1}^{i-1}\alpha_j+\gamma_j,  \\
\tilde{\delta}_i=\beta_i+\sum_{j=1}^{i-1}\beta_j+\delta_j,  \quad m_\pi := \tfrac 23 \sum_{i=1}^{\ell} ( \gamma_i + \beta_i). 
\end{gathered} 
\end{equation}
In \eqref{eq:form}, the sum is over elements of the finite set
\begin{equation}
\label{eq:theta_l}
\begin{split}\Theta_{\ell} &:=\Bigg\{\pi=\left[(\alpha_1,\beta_1),(\gamma_1,\delta_1),(\alpha_2,\beta_2),...,(\gamma_{\ell},\delta_{\ell})\right],\: \sum_{j=1}^{\ell}\alpha_j+\gamma_j=\sum_{j=1}^{\ell}\beta_j+\delta_j=0,\:
\\
&  \quad  (\alpha_i,\beta_i)\in \{(1,1),(-2,1),(1,-2)\},(\gamma_i,\delta_i)\in \{(-1,-1),(2,-1),(-1,2)\}
\Bigg\}.
\\
& 
\end{split} 
\end{equation}

Using \eqref{eq:form}, the diagonal part of ${A}_k^{\ell}$, is of the form
$$3^{\ell}\sum_{\pi\in \Theta_{\ell}}\omega^{m_{\pi}}\prod_{i=1}^{\ell}\Lambda'_{\alpha_i,\beta_i}\Lambda'_{\gamma_i,\delta_i}, \ \ \ \pi=\left[(\alpha_1,\beta_1),(\gamma_1,\delta_1),(\alpha_2,\beta_2)...,(\gamma_{\ell},\delta_{\ell})\right], $$
where $\Lambda'$ corresponds to the matrix where we only kept the coefficients $(n,m)$ where $(n,m)\in \left(3\mathbb{Z}+1\right)\times\left(3\mathbb{Z}+1\right)$ i.e
$$\Lambda'_{m,n}=\frac{1}{\omega^2(3m+1+k_1)-\omega(3n+1+k_2)}. $$ Theorem \ref{traceresult} then reduces to 
\begin{theo}
\label{theo:technical_tristan}
Let ${\ell} \ge 2$ and $\Theta_{\ell}$ be as in \eqref{eq:theta_l} with coefficients $\tilde \alpha_i,..,\tilde \delta_i, m_{\pi}$ as in \eqref{eq:coefficients}. 
Then the traces are given by
$${\tr}\left({A}_k^{\ell}\right)=-\frac{2i\omega\pi}{3}\sum_{\pi\in \Theta_{\ell}}\sum_{(\eta_i,\epsilon_i)\in \{ ({\tilde{\alpha}_i},{\tilde{\beta}_i}),({\tilde{\gamma}_i},{\tilde{\delta}_i}),1\leqslant i \leqslant l\}}{\Res}(f_{\pi},-\gamma_{(\eta_i,\epsilon_i)})\epsilon_i, $$
where with $\gamma_{(a,b)}=\omega^2a-\omega b$
$$f_{\pi}(k):=3^{\ell}\omega^{m_{\pi}}\prod_{i=1}^{\ell}\frac{1}{(k+\gamma_{({\tilde{\alpha}_i},{\tilde{\beta}_i})}+\mu)(k+\gamma_{({\tilde{\gamma}_i},{\tilde{\delta}_i})}+\mu)}, \quad \mu:=\omega^2-\omega.$$
\end{theo}
\begin{proof}
This is just a re-writing of formula of Theorem \ref{traceresult} in these rectangular coordinates. Indeed,  the $(0,0)$-th entry of the matrix $A_k^{\ell}$ is, in these notation, equal to $\sum_{\pi\in\Theta_{\ell}}f_{\pi}(k)$. Because we work in a the Hilbert space $L^2(\mathbb C/\Gamma,\mathbb C)$, this entry is also equal to $\langle A_k^{\ell}e_{i},e_{i}\rangle_{L^2}$. Now, the poles of this function are exactly described by $\gamma_{({\tilde{\alpha}_i},{\tilde{\beta}_i})}+\mu$ and $\gamma_{({\tilde{\gamma}_i},{\tilde{\delta}_i})}+\mu$ (this is a consequence of formula \ref{eq:form}) . Note however that in these coordinates, the poles get rescaled by $\sqrt 3$, this is why $\mu=-i\sqrt 3$ replaces $-i$. On the level of residues, this explains why a $\sqrt 3$ does not appear in this formula. Finally, we remark that in this decomposition, $\gamma_{({\tilde{\alpha}_i},{\tilde{\beta}_i})}+\mu\in \sqrt 3\big( 3\Gamma^*-i\big)$ and $\gamma_{({\tilde{\gamma}_i},{\tilde{\delta}_i})}+\mu\in \sqrt 3\big( 3\Gamma^*-2i\big)$ thus corresponding to the splitting appearing in the formula of stated in Theorem \ref{traceresult}.
\end{proof}

\smallsection{Acknowledgements} 
We would like to thank the two anonymous referees for their careful reading of the paper
and many helpful comments.
 TH and MZ  gratefully acknowledge partial support 
by the National Science Foundation under the grant DMS-1901462
and by the Simons Foundation Targeted Grant Award No.
896630.

\smallsection{Conflict of interests}
The authors have no competing interests to declare that are relevant to the content of this article.

\end{document}